\documentclass{article}
\usepackage{times}
\usepackage{fullpage,amsfonts,latexsym,graphicx,amssymb}
\usepackage{amsmath,amsthm,amstext,url,ifpdf}
\usepackage{ifpdf}
\ifpdf
\usepackage{hyperref}
\fi
\newcommand\Xcomment[1]{}

\newcommand{\eqdef}{\stackrel{def}{=}}

\newcommand{\Ex}{\mathbb{E}}

\newcommand{\e}{\epsilon}

\newcommand{\Remove}[1]{}
\newcommand{\one}{1}
\newcommand{\n}{n}
\newcommand{\two}{2}

\newcommand{\red}{white}
\newcommand{\comment}[1]{}

\newtheorem{theorem}{Theorem}[section]

\newtheorem{lemma}[theorem]{Lemma}
\newtheorem{corollary}[theorem]{Corollary}

\newcommand{\ud}{\mathrm{d}}
\begin{document}

\title{\Large  Balanced Allocations: A Simple Proof for the Heavily Loaded Case
}

\author{
Kunal Talwar\thanks{Microsoft Research, \texttt{kunal@microsoft.com}} \and
Udi Wieder\thanks{Microsoft Research, \texttt{uwieder@microsoft.com}}}


\maketitle


\begin{abstract}
We provide a relatively simple proof that the expected gap between the maximum load and the average load in the two choice process is bounded by $(1+o(1))\log \log n$, irrespective of the number of balls thrown. The theorem was first proven by Berenbrink et al. in \cite{BCSV00}. Their proof uses heavy machinery from Markov-Chain theory and some of the calculations are done using computers. In this manuscript we provide a significantly simpler proof that is not aided by computers and is self contained. The simplification comes at a cost of weaker bounds on the low order terms and a weaker tail bound for the probability of deviating from the expectation.
\end{abstract}




\newtheorem{definition}{Definition} 
\newtheorem{claim}{Claim} 
\newtheorem{remark}[definition]{Remark}
\newtheorem{example}[definition]{Example}

\Xcomment{
\newtheorem{notate}[definition]{Notation}
\newtheorem{theorem}[definition]{Theorem}
\newtheorem{facten}[definition]{Fact}
\newtheorem{problem}[definition]{Problem}
\newtheorem{corollary}[definition]{Corollary}
\newtheorem{obs}[definition]{Observation}
\newtheorem{notat}[definition]{Notation}
\newtheorem{lemma}[definition]{Lemma}
\newtheorem{remark}[definition]{Remark}

\newenvironment{proof}{{\bf Proof:} \rm}{\hfill $\square$\medskip}

\def\beginsmall#1{\vspace{-\parskip}\begin{#1}\itemsep-\parskip}
\def\endsmall#1{\end{#1}\vspace{-\parskip}}
}

\newcommand{\namedref}[2]{#1~\ref{#2}}
\newcommand{\sectionref}[1]{\namedref{Section}{#1}}
\newcommand{\appendixref}[1]{\namedref{Appendix}{#1}}
\newcommand{\subsectionref}[1]{\namedref{Subsection}{#1}}
\newcommand{\theoremref}[1]{\namedref{Theorem}{#1}}
\newcommand{\defref}[1]{\namedref{Definition}{#1}}
\newcommand{\figureref}[1]{\namedref{Figure}{#1}}
\newcommand{\figref}[1]{\namedref{Figure}{#1}}
\newcommand{\claimref}[1]{\namedref{Claim}{#1}}
\newcommand{\lemmaref}[1]{\namedref{Lemma}{#1}}
\newcommand{\tableref}[1]{\namedref{Table}{#1}}
\newcommand{\corollaryref}[1]{\namedref{Corollary}{#1}}
\newcommand{\propertyref}[1]{\namedref{Property}{#1}}
\newcommand{\appref}[1]{\namedref{Appendix}{#1}}
\newcommand{\propref}[1]{\namedref{Proposition}{#1}}




\section{A Bit of History}
In the Greedy$[d]$ process (sometimes called the $d$-choice process), balls are placed sequentially into $[n]$ bins with the following rule: Each ball is placed by uniformly and independently sampling $d$ bins and assigning the ball to the least loaded of the $d$ bins. In other words, the probability a ball is placed in one of the $i$ heaviest bins (at the time when it is placed) is exactly\footnote{Assume for simplicity and w.l.o.g that ties are broken according to some fixed ordering of the bins.}  $(i/n)^d$.  We remark that using this characterization there is no need to assume that $d$ is a natural number (though the process is algorithmically much simpler when $d$ is an integer). The main point is that whenever $d>1$ the process is \emph{biased}: the lighter bins have a higher chance of getting a ball. In this paper we are interested in the \emph{gap} of the allocation, which is the difference between the number of balls in the heaviest bin, and the average. The case $d=1$, when balls are placed uniformly at random in the bins, is well understood. In particular when $n$ balls are thrown the bin with the largest number of balls will have $\Theta(\log n/\log \log n)$ balls w.h.p. Since the average is $1$ this is also the gap. If $m>>n$ balls are thrown the heaviest bin will have  $m/n + \Theta(\sqrt{m\log n/n})$ balls w.h.p. \cite{RS98}.

In an influential paper Azar et al. \cite{ABKU99} showed that when $n$ balls are thrown and $d>1$ the gap is $\log \log n /\log d +O(1)$ w.h.p. The case $d=2$ is implicitly shown in Karp et al. \cite{KLH92}. The proof by Azar et al. uses a simple but clever induction; in our proof here we take the same approach. Bounding the number of balls by $n$ (or by $O(n)$) turns out to be a crucial assumption: the proof in \cite{ABKU99} breaks down once the number of balls is super-linear in the number of bins. Two other approaches to prove this result, namely, using differential equations or witness trees, also fail when the number of balls is large. See for example the survey \cite{MRS00}. A breakthrough was achieved by Berenbrink et al. in \cite{BCSV00}. They proved that the same bound on the gap holds for \emph{any}, arbitrarily large number of balls. Contrast this with the one choice case in which the gap diverges with the number of balls. At a (very) high level their approach was the following: first they show that the gap after $m$ balls are thrown is distributed similarly to the gap after only $poly(n)$ balls are thrown. This is done by bounding the mixing time of the underlying Markov Chain. The second step is to extend the induction technique of~\cite{ABKU99} to the case of $poly(n)$ balls. This turns out to be a major technical challenge which involves four inductive invariants and computer aided calculations. As such, finding a simpler proof remained an interesting open problem. In this paper we provide such a proof. The simplification comes at a minor cost: we get weaker tail bounds and higher lower order terms.  While~\cite{BCSV00} show that for any $c$, the gap is at most $\log\log n + \gamma(c)$ with probability $(1-\frac{1}{n^c})$ for a constant $\gamma(c)$ depending on $c$ alone, our proof shows that the gap is $\log\log n + \gamma'(c)\cdot\log\log\log n$ with probability $(1-\frac{1}{(\log\log n)^c})$ for a constant $\gamma'(c)$ depending on $c$ alone..

\section{The Proof}
We define the \emph{load vector} $X^t$ to be an $n$ dimensional vector where $X_i^t$ is the difference between the load of the $i$'th bin after $tn$ balls are thrown and the average $t$, (so that a load of a bin could be negative and $\sum X_i = 0 $). We also assume without loss of generality that the vector is sorted so that $X^t_\one\geq X^t_\two\geq...\geq X^t_{\n}$.
We will consider the Markov chain defined by  $X^{t}$, so one step of the chain consists of throwing $n$ balls according to the $d$-choice scheme and then sorting and normalizing the load vector.

The main tool we use is the following Theorem proven in \cite{PTW10} using a potential function argument. For the reader's convenience we include a proof in Section~\ref{sec:potential}.
\begin{theorem}
\label{thm:potentialbound}
There exists universal constants $a$ and $b$ which may depend on $d$ but not on $n$ or $t$, such that, $\Ex[\sum_i\exp(a|X_i^t|)]\leq bn$.
\end{theorem}

{
Let $G_t \eqdef X^t_\one$ denote the gap between maximum and average when sampling from $X^t$. Theorem \ref{thm:potentialbound} immediately implies the following:
\begin{lemma}\label{lem:logbound}
For any $t$, any $c\geq 0$, $\Pr[G^t \geq (c\log n)/a] \leq bn/n^c$. Thus for every $c$ there is a $\gamma = \gamma(c)$ such that $\Pr[G^t \geq \gamma \log n] \leq n^{-c}$.
\end{lemma}

Armed with this result, the crucial lemma, that we present next, says that if the gap at time $t$ is $L$, then after throwing another $nL$ balls, the gap becomes $\log\log n + O(\log L)$ with probability close to $1$. A bit more formally, if $L > \log_d\log n + O(\log\log\log n)$, the tail probabilities $\Pr[G^t\geq L]$ and $\Pr[G^{t+L} \geq \log_d\log n + O(\log L)]$ differ by at most an additive $\frac{1}{poly(L)} + \frac{1}{poly(n)}$.  Then using Lemma~\ref{lem:logbound}, we will infer a tail bound for $\Pr[G \geq \log_d\log n + O(\log\log\log n)]$.

\begin{lemma}\label{lem:inductionstep}
For any $c>0$ there is a $\gamma = \gamma(c)$, independent of $n$, so that for any $t,\ell,L$ such that $\ell \leq L\leq n^{\frac{1}{4}}$,
$\Pr[G^{t+L}\geq \log_d \log n + \ell + \gamma] \leq \Pr[G^t \geq L] + 8bL^3/\exp(a\ell) + n^{-c}$, where $a,b$ are the constants from Theorem~\ref{thm:potentialbound}.
\end{lemma}

The lemma is relatively straightforward to prove using the layered induction technique, except that we need a non-trivial ``base case'' to start the layered induction. Theorem~\ref{thm:potentialbound} provides us with such a base case, for bins with $\ell$ more balls than average in $X^{t+L}$. For a specific ball to increase the number of balls in a bin from $i$ to $i+1$, it must pick two bins that already contain at least $i$ balls. If the fraction of bins with at least $i$ balls when this ball is placed is at most $\beta_i$, then this probability would be $\beta_i^2$. While this $\beta_i$ value is a function of time, it is monotonically increasing and using the final $\beta_i$ value would give us an upper bound on the probability of increase. We get such a bound for the base case using our potential function bound, and use induction and Chernoff bounds to conclude that the gap is likely to be small. We next give the details of such an argument.

\begin{proof}
We sample an allocation $X^t$ and let $G^t$ be its gap. Now take an additional $L$ steps of the Markov chain to obtain $X^{t+L}$: in other words, we throw an additional $nL$ balls using the $d$-choice process. For brevity, we will use $X,G,X',G'$ to denote $X^t,G^t,X^{t+L},G^{t+L}$ respectively. We condition on $G < L$ and we prove the bound for $G'$. Let $L' = \log_d \log n + \ell + \gamma$. Observe that:
\begin{align}
\Pr[G' \geq L'] \leq  \Pr[G' \geq L'~|~ G < L] + \Pr[G \geq L]
\end{align}
It thus suffices to prove that $\Pr[G' \geq L'~|~ G < L] \leq 8bL^3/\exp(a\ell) + n^{-c}$. We do this using a layered induction similar to the one in~\cite{ABKU99}.

Let $\nu_i$ be the fraction of bins with load at least $i$ in $X'$, we will define a series of numbers $\beta_i$ such that $\nu_i \leq \beta_i$ with high probability. For convenience, let the balls existing in $X$ be black, and let the new $nL$ balls thrown be \red. We define the \emph{height} of a ball to be the load of the bin in which it was placed relative to $X'$. Let $\mu_i$ be the number of balls (out of the $nL$ \red\ balls thrown) that fall at height greater than $i$ in $X'$. Note that since a total of $nL$ \red\ balls are thrown, the average increases by $L$, so in order for a black ball to be in height $i$ in $X'$ it had to had been placed in a bin of load $L+i$ in $X$.  The main observation is that conditioned on $G < L$, no black ball is in a bin with load more than $L$ in $X$ and therefore all black balls are below the average of $X'$.  So, for any $i \geq 0$, it must be that $\nu_{i}n \leq \mu_i$.

\begin{figure}[t]
\centering
\includegraphics[scale=0.5]{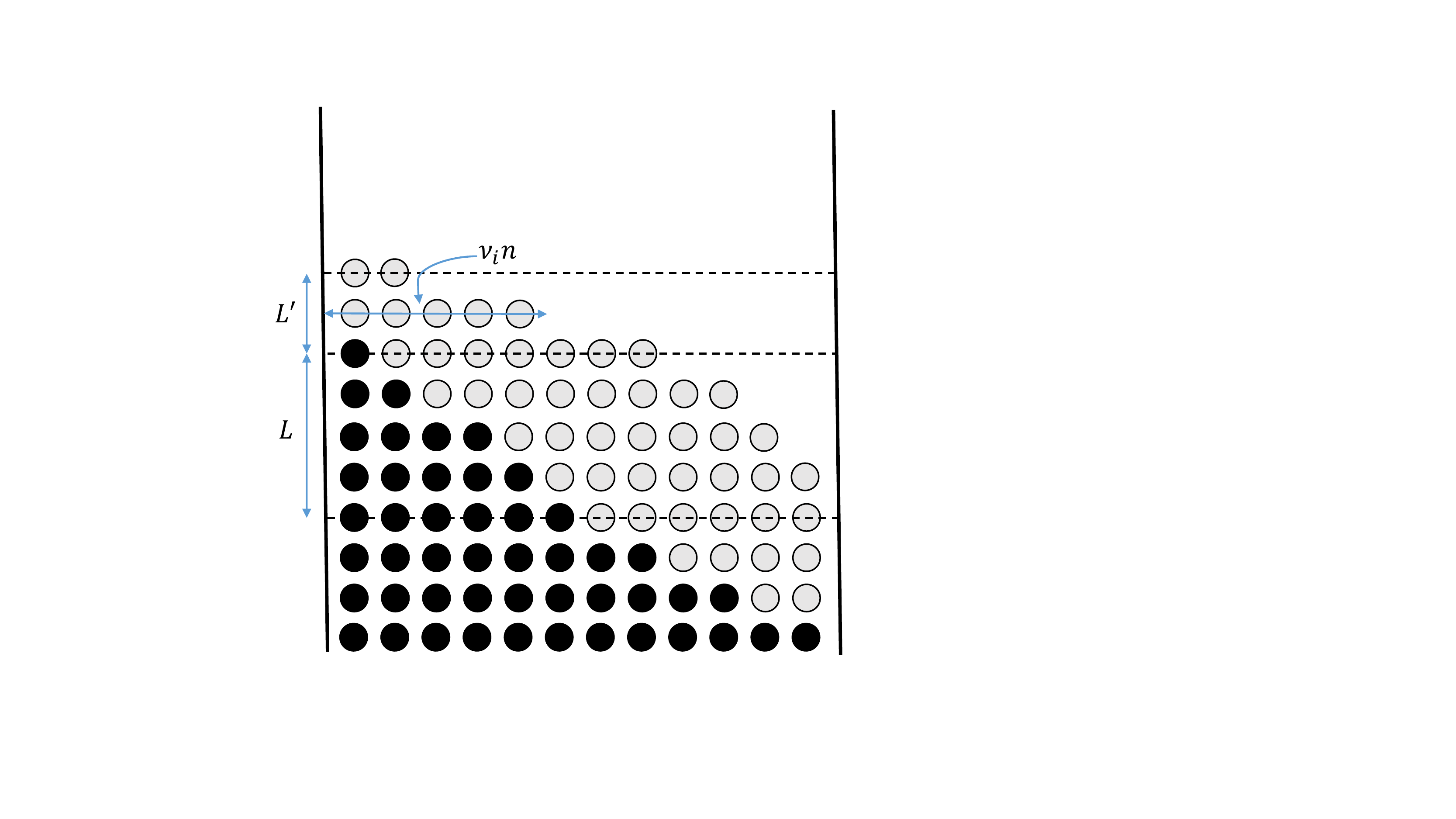}
\caption{Black balls are in $X$, $nL$ \red\ balls are thrown to obtain $X'$ }
\label{fig:allocation}
\end{figure}

By Theorem~\ref{thm:potentialbound} and Markov's inequality, $\Pr[\nu_{\ell} \geq \frac{1}{8L^{3/(d-1)}}] \leq \frac{8bL^{3/(d-1)}}{\exp(a\ell)}$, so we can set $\beta_\ell = \frac{1}{8L^{3/(d-1)}}$ as the base of the layered induction. By the standard layered induction argument we have that w.h.p $\nu_{i+1} \leq \mu_{i+1}/n \leq 2L\beta_i^d$ and so we set $\beta_{i+1}= 2L\beta_i^d$. Since $\beta_{i} < \frac{1}{8L^{3/(d-1)}}$ for $i \geq \ell$, the multiplicative term of $L$ has little impact, and we can derive the claimed bound. For completeness, we give details below. For ease of notation we assume $d=2$, the generalization for any $d>1$ is trivial.

\medskip

Let $i_L= \ell$, $i_H = i_L+\log\log n$ and $c'=3(c+1)$. Let  $\beta_{i_L} = \frac{1}{8L^3}$ and $\beta_{i+1}= \max(2L\beta_i^2,2c'\log n / n)$ for $i=i_L,\ldots,i_H-1$. It is easy to check that $\beta_{i_H} = 2c'\log n/n$. Indeed the recurrence
\begin{align*}
\log \beta_{i_L} &= -3\log (2L),\\
\log \beta_{i+1} &= 2\log \beta_{i} + \log (2L)
\end{align*}
solves to $\log \beta_{i_L+k} = \log (2L) (-3\cdot 2^k + (2^k-1))$, which implies the claim. The inductive step in the layered induction is the following:
\begin{lemma}
For any $i \in [i_L,i_H-1]$, we have $\Pr[\nu_{i+i} > \beta_{i+1} \mid \nu_i \leq \beta_i] \leq \frac{1}{n^{c+1}}$.
\end{lemma}
\begin{proof}
For a ball to fall at height at least $i+1$, it should pick two bins that have load at least $i$ when the ball is placed, and hence at least as much in $X'$. Thus the probability that a ball falls at height at least $i+1$  is at most $\nu_{i}^2 \leq \beta_i^2$ under our conditioning. Since we place $nL$ balls, the expected number of balls that fall at height at least $i+1$ is bounded by  $nL\beta_i^2 \leq n\beta_{i+1}/2$. Finally, since this number is at least $c'\log n$, Chernoff bounds imply that the probability that we get twice the expectation is at most $\exp(-c'\log n/3) \leq 1/n^{c+1}$. The claim follows.
\end{proof}

It follows that $\Pr[\nu_{i_H} > \beta_{i_H}~|~ G < L] \leq 8bL^3/\exp(a\ell) + i_H/n^{c+1}$. Now we condition on $\nu_{i_H} \leq \beta_{i_H}$, and let $H$ be the set of bins of height at least $i_H$ in $X'$. Once a bin reaches this height, an additional ball falls in it with probability at most $(2\beta_{i_H} n+1)/n^2$. Thus the expected number of balls falling in such a bin is $O(L\log n / n)$. The probability that any bin in $H$ gets $2c$ balls after reaching height $i_H$ is then at most $O(\log n \exp(-\Omega(c^2n/3L\log n)) \leq 1/n^{c+1}$ for large enough $n$. The claim follows.
\end{proof}

This lemma allows us to bound $\Pr[G^{t+L} \geq \log\log n + O(\log L) ]$ by $\Pr[G^t \geq L] + \frac{1}{poly(L)}$. Since $\Pr[G^t \geq O(\log n)]$ is small, we can conclude that $\Pr[G^{t+O(\log n)}\geq O(\log\log n)]$ is small. Another application of the lemma then gives that $\Pr[G^{t+O(\log n) + O(\log\log n)} \geq \log\log n + O(\log\log\log n)]$ is small. We formalize these corollaries next.
\begin{corollary}
\label{cor:logscale}
There is a universal constant $\gamma$ such that for any $k \geq 0$, $t \geq (12\log n)/a$, $\Pr[G^{t} \geq (5+\frac{10}{a})\cdot \log\log n + k + \gamma] \leq \frac{1}{n^{10}}+\frac{\exp(-ak)}{\log^4 n} $.
\end{corollary}
\begin{proof}
Set $L=12\log n / a$, and use Lemma~\ref{lem:logbound} to bound $\Pr[G^{t-L}\geq L]$. Set $\ell = k+\log(8bL^3\log^4 n)/a$ to derive the result.
\end{proof}
\begin{corollary}
\label{cor:loglogscale}
There are universal constants $\gamma,\alpha$ such that for any $k\geq 0$, $t \geq \omega(\log n)$, $\Pr[G^t \geq \log\log n + \alpha\log\log\log n + k + \gamma] \leq \frac{1}{n^{10}}+\frac{1}{\log^4 n}+\frac{\exp(-ak)}{(\log\log n)^4}$.
\end{corollary}
\begin{proof}
Set $L=\log(8b(\frac{12\log n}{a})^3 \log^4 n)/a = \frac{7\log\log n}{a} + O_{a,b}(1)$ and use Corollary~\ref{cor:logscale} with $k$=0 to bound $\Pr[G^{t-L}\geq L]$. Set $\ell = k+\log(8bL^3(\log\log n)^4)/a$ to derive the result.
\end{proof}

Setting $k=0$ in Corollary~\ref{cor:loglogscale}, we conclude that
\begin{corollary}
There are universal constants $\gamma,\alpha$ such that for $t \geq \omega(\log n)$, $\Pr[G^t \geq \log\log n + \alpha\log\log\log n + \gamma] \leq \frac{2}{(\log\log n)^4}$.
\end{corollary}

Using the above results, we can also conclude
\begin{corollary}
There are universal constants $\gamma,\alpha$ such that for $t \geq \omega(\log n)$ $\Ex[G^t] \leq \log\log n + \alpha\log\log\log n + \gamma$.
\end{corollary}
\begin{proof}
Let $\ell_1= \log\log n + \alpha\log\log\log n + \gamma_1$ for $\alpha,\gamma_1$ from Corollary~\ref{cor:loglogscale}, and let $\ell_2 =  (5+\frac{10}{a})\cdot \log\log n+\gamma_2$ for $\gamma_2$ from Corollary~\ref{cor:logscale}. Finally, let $\ell_3 = 12\log n/a$. We bound
\begin{align*}
\Ex[G^t] &\leq \ell_1 + \int_{\ell_1}^{\ell_2} \Pr[G^t \geq x] \,\ud x + \int_{\ell_2}^{\ell_3} \Pr[G^t\geq x] \,\ud x + \int_{\ell_3}^{\infty} \Pr[G^t \geq x] \,\ud x\\
\end{align*}
Each of the three integrals are bounded by constants, using Corollaries~\ref{cor:loglogscale} and \ref{cor:logscale} and Lemma~\ref{lem:logbound} respectively. The claim follows.
\end{proof}

The following lemma states that the lower bound condition on $t$ is unnecessary.

\begin{lemma}\label{lem:major}
For $t \geq t'$, $G^{t'}$ is stochastically dominated by $G^t$. Thus $\Ex[G^{t'}] \leq \Ex[G^t]$ and for every $k$,  $\Pr[G^{t'}\geq k] \leq \Pr[G^t \geq k]$.
\end{lemma}
\begin{proof}[Proof sketch.]
We use the notion of majorization, which is a variant of stochastic dominance. See for example \cite{ABKU99} for definitions. Observe that trivially $X^0$ is majorized by $X^{t-t'}$. Now throw $nt'$ balls using the standard coupling and get that $X^{t'}$ is majorized by $X^{t}$. The definition of majorization implies the stochastic dominance of the maximum and the bounds on the expectation and the tail follow.
\end{proof}
}
\section{Extensions}
The technique we use naturally extends to other settings.
\subsection{The Weighted Case}
Previously we assumed all balls are of unit weight. For the case of varying weights we use the model proposed in \cite{TW07} and also used in \cite{PTW10}. Every ball comes with a weight $W$ independently sampled from a weight distribution $\mathcal W$. Without loss of generality we assume $\Ex[\mathcal W]=1$. The weight of a bin is the sum of weights of balls assigned to it. The \emph{gap} is naturally defined as the difference between the weight of the heaviest bin and the average bin. In \cite{TW07} it is shown that if $\mathcal W$ has a bounded second moment and satisfies some additional mild smoothness condition, then the expected gap does not depend on the number of balls. The paper does not provide any explicit bounds on the gap though. In \cite{PTW10} it is shown that if $\mathcal W$ has a finite exponential generating function the gap is bounded by $O(\log n)$. For some distributions, such as the exponential distribution, this bound is tight. Here we can show that if $\mathcal W$ is very concentrated (for instance it is bounded) then better bounds can be proved.

Consider for example the case where the size of each ball is drawn uniformly from $\{1,2\}$. Previous techniques such as~\cite{BCSV00} fail to prove an $O(\log\log n)$ bound in this case, and the best bound prior to this work is the $O(\log n)$ via the potential function argument of~\cite{PTW10}. The fact that Theorem~\ref{thm:potentialbound} holds means that the technique of this paper can be applied. Moreover, the layered induction still works if we go up in steps of size two instead of one. This shows a bound of $2\log_d\log n + O(1)$ for this distribution.

More generally, for a weight distribution $W$ with a bounded exponential moment generating function, let $M_s$ be the smallest value such that $\Pr[W \geq M_s] \leq \frac{1}{s(\log\log n)^5}$. Then a proof analogous to Lemma~\ref{lem:inductionstep} shows that the gap is $O(\log\log n) + \sum_{i=i_L}^{i_H} M_{\beta_{i}n}$. If $M_n$ is $\omega(\log\log n)$, then this is $O(M_n)$, which is tight up to constants. We note however that this proof leaves a ``hole'': since majorization does not necessarily hold in the weighted case, our approach proves the bound on the gap when $\Omega(n\log n)$ balls are thrown.
\subsection{The Left$[d]$ Scheme}
Next we sketch how this approach also proves a tight bound for V\"{o}cking's Left$[d]$ process~\cite{Voc03}. The result had been shown in \cite{BCSV00}, though there they had to redo large sections of the proof (and the most technical at that), while here we only require minor changes. Recall that in Left$[d]$ the bins are partitioned into $d$ sets of $n/d$ bins each (we assume $n$ is divisible by $d$). When placing a ball,  one bin is sampled uniformly from each set and the ball is placed in the least loaded of the $d$ bins. The surprising feature of this process is that ties are broken according to a fixed ordering of the sets (we think of the sets as ordered from left to right and ties are broken 'to the left', hence the name of the scheme). The surprising result is that the gap now drops from $\frac{\log\log n}{\log d}$ to $\frac{\log \log n}{d\ln\phi_d}$ where $\phi_d=\lim_{k\rightarrow\infty}(F_d(k))^{\frac{1}{k}} \in [1.61,2)$ is the base of the order $d$ Fibonacci number.

The key ingredient in the proof is Theorem \ref{thm:potentialbound} from \cite{PTW10}. The exponential potential function is Schur-Convex and therefore the theorem holds for any process which is majorized by the Greedy$[d]$ process. It is indeed the case that V\"{o}cking's Left$[d]$ process~\cite{Voc03} is majorized by Greedy$[d]$ (see the proof in~\cite{BCSV00}). All that remains is to prove the analog of Lemma~\ref{lem:inductionstep}. For this we follow the analysis of Mitzenmacher and V\"{o}cking in~\cite{MV98}.
Let $X_{jd+k}$ be the number of bins of load at least $j$ from the $k$'th set, and set $x_i = X_i/n$. It is easy to verify the recursive equation
$$\Ex[x_i | x_{< i}] \leq d^d \prod _{j=i-d}^{i-1} x_j$$
From here the proof is similar to that of Lemma~\ref{lem:inductionstep}.

\section{Discussion}

The theorem in \cite{BCSV00} states that for every $c$ there is a $\gamma = \gamma(c)$ so that $\Pr[G>\log \log n  + \gamma]\leq n^{-c}$. The reason  our techniques do not show such a sharp bound is that we do not obtain a small enough tail for the base case of the layered induction, i.e.  on  $\Pr[\nu_\ell \leq \beta_{\ell}]$. The reason is that the exponential potential function in  Theorem~\ref{thm:potentialbound} is not concentrated enough to yield such a bound. This presents a substantial obstacle, it seems that a different technique is needed in order to recover the results in \cite{BCSV00} at full strength.

\medskip

An interesting corollary from Theorem \ref{thm:potentialbound} is that the Markov chain  $X^t$ has a stationary distribution and that the bounds we prove hold also for the stationary distribution itself. In that sense, while in \cite{BCSV00} the mixing of the chain was used to move the interesting events to be closer to the "present", in our technique we allow ourselves to look directly at the distant "future". When balls are unweighted a simple majorization based argument shows that moving closer in time can only improve the bounds on the gap (this is Lemma~\ref{lem:major}). Unfortunately, a similar Lemma does not hold when balls are weighted (see \cite{BFHR08}), so we need to be specify the time periods we look at. Indeed, while our results hold when considering a large number of balls, we have a 'hole' for a number of balls that is smaller than $n\log n$.



\appendix
\section{Potential Function}\label{sec:potential}
In order to make the writeup self contained we next provide a proof of Theorem~\ref{thm:potentialbound}.

It would be convenient to define the load vector $x(t)$ to the sorted vector of gaps after $t$ balls are thrown, where $t$ is not necessarily a multiple of $n$, as in the previous section. In other words, $x(t)_i$ is the difference between the number of balls in the $i$th most loaded bin and the average $t/n$. Note that in the notation of the previous section, $X^t$ is $x(nt)$. The load of a bin now is not necessarily an integer. We define $p_i$ to be the probability the $i$'th loaded bin receives a ball, so $p_i = \left(\tfrac{i}{n}\right)^d - \left(\tfrac{i-1}{n}\right)^d$. Recall that we also have a weight distribution $\mathcal W$. The Markov chain is thus the following:

\begin{itemize}
\item sample $j \in_{\mathbf p} [n]$, i.e. pick $j$ with probability $p_j$.
\item sample $W \in \mathcal{W}$
\item set $y_i = x(t)_i + W - \tfrac{W}{n}$ for $i=j$ and  $z_i = x(t)_i - \tfrac{W}{n}$ for $i\neq j$
\item obtain $x(t+1)$ by sorting $y$
\end{itemize}

We make the following two observations which hold whenever $d>1$. It turns out to be all we need:

\begin{description}
\item \begin{align}\label{assume1}p_i \leq p_{i+1} \mbox{ for } i\in [n-1]\end{align}

\item For some $\e>0$ it holds that
\begin{align}
\label{assume2}
\sum_{i\geq \frac{3n}{4}} p_i \geq \tfrac{1}{4} + \e \;\;&\text{ and } \sum_{i\leq \frac{n}{4}} p_i \leq \tfrac{1}{4} - \e\end{align}
\end{description}

For the distribution $\mathcal W$, we assume that there is a $\lambda >0$ such that the moment generating function $ M[\lambda] = \Ex[e^{\lambda W}] < \infty$. Further, without loss of generality, $\Ex[W]=1$.
Note that
$$M''(z) = \Ex[W^2 e^{z W}] \leq \sqrt{\Ex[W^4]\Ex[e^{2z W}]}.$$
The above assumption implies that there is an $S\geq 1$ such that for every $|z| < \lambda/2$ it holds that  $M''(z) < 2S$. For simplicity, we assume throughout that $n$ is bounded below by a large enough constant.

Let $\alpha =\min(\tfrac{\e}{6S}, \lambda/2)$. We can assume that $\e\leq 1/4$ and thus that $\alpha \leq 1/6$. Define the following potential functions
\begin{align*}
 \Phi(x(t)) &:= \sum_{i=1}^{n} \exp(\alpha x(t)_i)\\
 \Psi(x(t))  &:= \sum_{i=1}^{n} \exp(-\alpha x(t)_i)\\
 \Gamma(x(t)) &:=  \Phi(x(t)) + \Psi(x(t))
\end{align*}

We start by calculating the expected change of $\Phi$ and $\Psi$ individually.
For ease of notation we write $\Phi$  or $\Phi(t)$ when the context clear.

\begin{lemma} For $\Phi$ defined as above,
\begin{equation}\label{BPhi} \Ex[\Phi(t+1) - \Phi(t)~|~x(t) ]  \leq \sum_{i=1}^n \left(p_i(\alpha +S\alpha^2) - (\tfrac{\alpha}{n}-S\tfrac{\alpha^2}{n^2})\right) e^{\alpha x_i}.
\end{equation}
\end{lemma}
\begin{proof}
Let $\Delta_i$ denote the change in $\Phi_i= \exp(\alpha x_i)$, i.e. $\Delta_i = \exp(\alpha y_i) -  \exp(\alpha x_i)$, where $y_i = x_i + W - \frac{W}{n}$ with probability $p_i$, and $y_i= x_i -\frac{W}{n}$ otherwise. In the first case, when the ball is placed in bin $i$, the expected change (taken over randomness in $W$) $\Delta_i$ is
\begin{eqnarray*}
\Ex[e^{\alpha(x_i + W-\frac{W}{n})}] - e^{\alpha x_i} &=& e^{\alpha x_i}(M(\alpha (1-\tfrac{1}{n})) -1 )\\
&=&e^{\alpha x_i} ( M(0) + M'(0)\alpha (1-\tfrac{1}{n}) + M''(\zeta) (\alpha (1-\tfrac{1}{n}))^2/2 - 1)
\end{eqnarray*}
for some $\zeta \in [0, \alpha (1-\tfrac{1}{n})]$. By the assumption on $\mathcal{W}$ and $\alpha$, $M''(\zeta) \leq 2S$. Moreover, $M(0) = 1$ and $M'(0) = \Ex[W]=1$. Thus the above expression can be bounded from above by
$$ e^{\alpha x_i} ( \alpha (1-\tfrac{1}{n}) + S \alpha^2).$$

Similarly, in the case that the ball goes to a bin other than $i$, the expected value of $\Delta_i$ can be bounded by $(-\tfrac{\alpha}{n}+ S\tfrac{\alpha^2}{n^2})e^{\alpha x_i}$.
Thus
$$\Ex[\Delta_i] \leq    p_i(\alpha(1-\tfrac{1}{n}) + S\alpha^2)e^{\alpha x_i} - (1-p_i)(\tfrac{\alpha}{n}-S\tfrac{\alpha^2}{n^2})e^{\alpha x_i} \leq  \left(p_i(\alpha + S\alpha^2) - (\tfrac{\alpha}{n}-S\tfrac{\alpha^2}{n^2})\right)e^{\alpha x_i}.$$
The claim follows.
%
%
%
%
%
%
%
 \end{proof}

\begin{corollary}
\label{cor:phiincrease}
\begin{align}\Ex[\Phi(t+1) - \Phi(t)~|~x(t) ] \leq \frac{2\alpha}{n} \Phi(t)\end{align}
\end{corollary}
\begin{proof}
Note that $S\alpha \leq \frac{1}{6} < 1$ so that
\begin{align}\Ex[\Phi(t+1) - \Phi(t)~|~x(t) ] \leq \sum_{i=1}^n 2\alpha p_i e^{\alpha x_i}.\end{align}
The claim follows by observing that $p_i$'s are increasing and $x_i$'s are decreasing, so that the expression is at most what it would be if the $p_i$'s were all equal.
\end{proof}


Similar arguments show that
\begin{lemma}
Let $\Psi$ be defined as above. Then
\begin{equation}\label{BPsi} \Ex[\Psi(t+1) - \Psi(t)~|~x(t) ]  \leq \sum_{i=1}^n \left(p_i(-\alpha +S\alpha^2)  + (\tfrac{\alpha}{n}+S\tfrac{\alpha^2}{n^2})\right)e^{-\alpha x_i}.
\end{equation}
\end{lemma}
\begin{corollary}
\label{cor:psiincrease}
\begin{align}\Ex[\Psi(t+1) - \Psi(t)~|~x(t) ] \leq \frac{2\alpha}{n} \Psi(t)\end{align}
\end{corollary}
\begin{proof}
This follows immediately as $p_i>0$ and $S\alpha < \frac{1}{6}$.\end{proof}

We start by showing that for reasonably balanced configurations, both $\Phi$ and $\Psi$ have the right decrease in expectation. More precisely, if $x_{\frac{3n}{4}} \leq 0$, then $\Phi$ decreases in expectation, and if $x_{\frac{n}{4}} \geq 0$, then $\Psi$ decreases in expectation.

\begin{lemma}
\label{lem:phidecreases}
Let $\Phi$ be defined as above. If $x_{\frac{3n}{4}}(t) \leq 0$, then $\Ex[\Phi(t+1)~|~x(t)] \leq (1-\frac{\alpha\e}{n})\Phi(t) + 1 $.
\end{lemma}
\begin{proof}
We upper bound $\sum_{i=1}^n p_i(\alpha+S\alpha^2) e^{\alpha x_i}$ for a fixed $\Phi(x)$, for $x$ which is non increasing with $\sum_i x_i =0$. We first write
\begin{eqnarray}
\sum_{i=1}^n p_i(\alpha + S\alpha^2) e^{\alpha x_i}\nonumber &\leq& \sum_{i < \frac{3n}{4}} p_i(\alpha + S\alpha^2) e^{\alpha x_i} + \sum_{i \geq \frac{3n}{4}} p_i(\alpha + S\alpha^2)e^{0} \nonumber \\
&\leq& \sum_{i < \frac{3n}{4}} p_i(\alpha + S\alpha^2) e^{\alpha x_i} + 1 \label{b1}
\end{eqnarray}
since $\alpha + S\alpha^2 \leq \tfrac{6\e + \e^2}{36S} \leq 1$ by our assumptions that $\e\leq 1$ and $S\geq 1$.

Now set $y_i := e^{\alpha x_i}$. The first term above is no larger than the maximum value of
\begin{align*}
&(\alpha + S\alpha^2) \sum_{i < \frac{3n}{4}} p_i y_i\\
&~~~\text{subject to}\\
&\sum_{i < \frac{3n}{4}} y_i \leq \Phi\\
&y_{i-1} \geq y_{i} ~~~ \forall ~1<i<\tfrac{3n}{4}.
\end{align*}
Since $\mathbf{p}$ is non-decreasing and $\mathbf{y}$ is non-increasing, the maximum is achieved when $y_i = \tfrac{4\Phi}{3n}$ for each $i$, and is at most $(\alpha + S\alpha^2)(\frac{3}{4}-\e)\tfrac{4\Phi}{3n}$. 

We can now plug this bound in \eqref{b1}, and substituting in \eqref{BPhi} we upper-bound the expected change in $\Phi$. 
\begin{eqnarray*}
\Ex\left[\Phi(t+1)-\Phi(t)~|~x(t)\right]
&\leq& (\alpha + S\alpha^2)(\frac{3}{4}-\e)\tfrac{4\Phi}{3n}  - \left(\frac{\alpha}{n}-S\frac{\alpha^2}{n^2}\right)\Phi + 1\\
&\leq& \frac{\alpha\Phi}{n} \big( (1+S\alpha)(1-\frac{4\e}{3}) - 1  + S\frac{\alpha}{n}\big) + 1\\
~\text{Assuming } S\alpha \leq \e/6  \text{ we have}&&\\
&\leq& \frac{\alpha}{n}\Phi \left(\tfrac{\e}{6}-\tfrac{4\e}{3}+ \tfrac{\e}{6n} \right)+1\\
&\leq& -\frac{\alpha\e}{n} \Phi + 1
\end{eqnarray*}
The claim follows.
\end{proof}

\begin{lemma}
\label{lem:psidecreases}
Let $\Psi$ be defined as above. If $x_{\frac{n}{4}}(t) \geq 0$, then $\Ex[\Psi(t+1)~|~x(t)] \leq (1-\frac{\alpha\e}{n})\Psi(t)+1$.
\end{lemma}
\begin{proof} 
We first upper bound  $\sum_{i=1}^n p_i(-\alpha+S\alpha^2) e^{-\alpha x_i}$ for a fixed $\Psi(x)$, for $x$ which is non increasing with $\sum_i x_i =0$. Since $(-\alpha+S\alpha^2)$ is negative, we have
\begin{align*}
\sum_{i=1}^n p_i(-\alpha + S\alpha^2) e^{-\alpha x_i} &\leq (-\alpha + S\alpha^2) \sum_{i \geq \frac{n}{4}} p_i e^{-\alpha x_i}
\end{align*}

Now set $z_i := e^{-\alpha x_i}$. Under the assumption on $x_{\frac{n}{4}}$, the sum $\sum_{i \geq \frac{n}{4}} z_i$ is at least $\Psi - \frac{n}{4}$. Since $(-\alpha+S\alpha^2)$ is negative, to upper bound the second term, we need to find the minimum value of
\begin{align*}
&\sum_{i \geq \frac{n}{4}} p_i z_i\\
&~~~\text{subject to}\\
&\sum_{i \geq \frac{n}{4}} z_i \geq \Psi - \frac{n}{4}\\
&z_{i-1} \geq z_{i} ~~~ \forall ~i>\tfrac{n}{4}.
\end{align*}
Since both $\mathbf{p}$ and $\mathbf{z}$ are (weakly) increasing, the minimum is achieved when $z_i = \tfrac{4(\Psi - \frac{n}{4})}{3n}$ for each $i$. Using the assumption that $\sum_{i\geq n/4}p_i \geq \frac{3}{4}+\e$ we can bound the expression above  by  $(-\alpha + S\alpha^2)(\frac{3}{4}+\e)\tfrac{4(\Psi - \frac{n}{4})}{3n}$. We can now upper-bound the expected change in $\Psi$ by plugging this bound in \eqref{BPsi}. 
\begin{eqnarray*}
\Ex[\Psi(t+1)-\Psi(t)~|~x(t)]
&\leq& (-\alpha + S\alpha^2)(\tfrac{3}{4}+\e)\tfrac{4(\Psi - \frac{n}{4})}{3n}  + \tfrac{\alpha}{n}(1+S\tfrac{\alpha}{n}) \Psi\\
&=& \tfrac{\alpha }{n}\left((1+S\tfrac{\alpha}{n})\Psi+(-1+S\alpha)(\tfrac{3}{4}+\e)\tfrac{4\Psi -n}{3}\right)   \\
&=& \tfrac{\alpha }{n}\left((1+S\tfrac{\alpha}{n})\Psi+S\alpha(\tfrac{3}{4}+\e)\tfrac{4\Psi -n}{3}-(\tfrac{3}{4}+\e)\tfrac{4\Psi -n}{3}\right)   \\
&\leq& \tfrac{\alpha\Psi}{n} \big(1+S\tfrac{\alpha}{n}+  S\alpha(\tfrac{3}{4}+\e)\tfrac{4}{3}- (\tfrac{3}{4}+\e)\tfrac{4}{3}\big) + \frac{\alpha}{3}(\tfrac{3}{4}+\e)\\
&\leq& -\frac{\alpha\e}{n} \Psi + 1
\end{eqnarray*}
where the last inequality follows since $\e\leq  \tfrac{1}{4}$ and $S\alpha \leq \tfrac{\e}{6}$.

%
\end{proof}

The next lemma will be useful in the case that $x_{\frac{3n}{4}} > 0$.
\begin{lemma}
\label{lem:phibadcase}
Suppose that $x_{\frac{3n}{4}} > 0$ and $\Ex[\Delta\Phi|x(t)] \geq -\frac{\alpha\e}{4n}\Phi$. Then either $\Phi < \frac{\e}{4}\Psi$ or $\Gamma < cn$ for some $c=poly(\frac{1}{\e})$.
\end{lemma}
\begin{proof}
First note that the expected increase in $\Phi$ is at most
\begin{eqnarray}
\sum_{i} (p_i(\alpha + S\alpha^2) - \tfrac{\alpha}{n} + S\tfrac{\alpha^2}{n^2}) e^{\alpha x_i}\nonumber
&\leq& \sum_{i\leq n/3} (p_i(\alpha + S\alpha^2) - \tfrac{\alpha}{n}+ S\tfrac{\alpha^2}{n^2}) e^{\alpha x_i} \nonumber
+ (\alpha + S\alpha^2)\sum_{i>n/3}p_ie^{\alpha x_i}\\
&\leq&  -\frac{\alpha \e}{2n}\Phi_{\leq n/3} +  \frac{2\alpha}{n}\Phi_{>n/3} \nonumber\\
&\leq&  -\frac{\alpha \e}{2n}\Phi +  \frac{3\alpha}{n}\Phi_{>n/3}\label{maxphi}
\end{eqnarray}
where in the next to last inequality we used  that for $i\leq n/3$, $p_i \leq \tfrac{1-4\e}{n}$ and  that for given $\Phi$, $\sum p_i e^{\alpha x_i}$ is maximized when $\mathbf p$ is uniform.

Thus $\Ex[\Delta\Phi|x(t)] \geq -\frac{\alpha\e}{4n}\Phi$  implies that
$$\frac{3\alpha}{n} \Phi_{> \frac{n}{3}} \geq \frac{\alpha\e}{4n} \Phi.$$
Let $B = \sum_{i} \max(0,x_i) = \frac{1}{2} ||x||_{1}$. Note that $\Phi_{\geq \frac{n}{3}}$ is upper bounded by $\frac{2n}{3} e^{\frac{3\alpha B}{n}}$. Thus
\begin{equation}
\Phi \leq \frac{12}{\e} \Phi_{> \frac{n}{3}} \leq \frac{8n}{\e} e^{\frac{3\alpha B}{n}}.
\label{eq:phiub}
\end{equation}
On the other hand, $x_{\frac{3n}{4}} > 0$ implies that $\Psi \geq \frac{n}{4} e^{\frac{4\alpha B}{n}}$.

If $\Phi < \frac{\e}{4}\Psi$, we are already done. Otherwise,
\begin{align*}
\frac{8n}{\e} e^{\frac{3\alpha B}{n}} \geq \Phi \geq \frac{\e}{4}\Psi \geq \frac{\e n}{16} e^{\frac{4\alpha B}{n}}
\end{align*}
so that $e^{\frac{\alpha B}{n}} \leq \frac{128}{\e^2}$. It follows that $$\Gamma \leq \tfrac{5}{\e}\Phi \leq \tfrac{40n}{\e^2}(\tfrac{128}{\e})^3 \leq cn.$$
\end{proof}

Similarly,
\begin{lemma}
\label{lem:psibadcase}
Suppose that $x_{\frac{n}{4}} < 0$ and $\Ex[\Delta\Psi|x(t)] \geq -\frac{\alpha\e}{4n}\Psi$. Then either $\Psi < \frac{\e}{4}\Phi$ or $\Gamma < cn$ for some $c=poly(\frac{1}{\e})$.
\end{lemma}
\begin{proof}
First observe that for any $i > \frac{2n}{3}$, $p_i > \frac{1+\e}{n}$ so that $p_i(-\alpha +S\alpha^2)+(\tfrac{\alpha}{n}+S\tfrac{\alpha^2}{n^2})\leq -\frac{\alpha\e}{2n}$. Since $p_i \geq 0$  it holds that  $p_i(-\alpha +S\alpha^2)+(\tfrac{\alpha}{n}+S\tfrac{\alpha^2}{n^2}) \leq \frac{2\alpha}{n}$ for every $i$. Using the upper bound from \eqref{BPsi} we get
\begin{align*}
\Ex[\Delta \Psi~|~x(t)] &\leq -\frac{\alpha\e}{2n} \Psi_{>\frac{2n}{3}} + \frac{2\alpha}{n} \Psi_{\leq \frac{2n}{3}}\\
&= -\frac{\alpha\e}{2n} \Psi + \frac{4\alpha  + \alpha\e}{2n} \Psi_{\leq \frac{2n}{3}}\\
&\leq -\frac{\alpha\e}{2n} \Psi + \frac{3\alpha}{n} \Psi_{\leq \frac{2n}{3}}.
\end{align*}

Thus $\Ex[\Delta\Psi~|~x(t)] \geq -\frac{\alpha\e}{4n}\Psi$  implies that
$$\frac{3\alpha}{n} \Psi_{\leq \frac{2n}{3}} \geq \frac{\alpha\e}{4n} \Psi.$$
Let $B = \sum_{i} \max(0,x_i) = \frac{1}{2} ||x||_{1}$. Note that $\Psi_{\leq \frac{2n}{3}}$ is upper bounded by $\frac{2n}{3} e^{\frac{3\alpha B}{n}}$. Thus
\begin{equation}
\Psi \leq \frac{12}{\e} \Psi_{\leq \frac{2n}{3}} \leq \frac{8n}{\e} e^{\frac{3\alpha B}{n}}.
\label{eq:psiub}
\end{equation}
On the other hand, $x_{\frac{n}{4}} < 0$ implies that $\Phi \geq \frac{n}{4} e^{\frac{4\alpha B}{n}}$.

If $\Psi < \frac{\e}{4}\Phi$, we are already done. Otherwise,
\begin{align*}
\frac{8n}{\e} e^{\frac{3\alpha B}{n}} \geq \Psi \geq \frac{\e}{4}\Phi \geq \frac{n \e}{16} e^{\frac{4\alpha B}{n}}
\end{align*}
so that $e^{\frac{\alpha B}{n}} \leq \frac{128}{\e^2}$. It follows that $$\Gamma \leq \tfrac{5}{\e}\Psi \leq \tfrac{40n}{\e^2}(\tfrac{128}{\e})^3 \leq cn.$$
\end{proof}

We are now ready to prove the supermartingale-type property of $\Gamma$.

\begin{theorem}
\label{thm:gammadecreases}
Let $\Gamma$ be as above. Then $\Ex[\Gamma(t+1)~|~x(t)] \leq (1-\frac{\alpha\e}{4n})\Gamma(t) + c$, for a constant $c= c(\e) = poly(\frac{1}{\e})$.
\end{theorem}
\begin{proof}

The proof proceeds via a case analysis. In case the conditions, $x_{\frac{n}{4}} \geq 0$ and $x_{\frac{3n}{4}} \leq 0$ hold, we show both $\Phi$ and $\Psi$ decrease in expectation. If one of these is violated Lemmas~\ref{lem:phibadcase} and~\ref{lem:psibadcase} come to the rescue.

\medskip\noindent{\bf Case 1: $x_{\frac{n}{4}} \geq 0$ and $x_{\frac{3n}{4}} \leq 0$.} In this case the theorem follows from Lemmas~\ref{lem:phidecreases} and~\ref{lem:psidecreases}.

\medskip\noindent{\bf Case 2: $x_{\frac{n}{4}} \geq x_{\frac{3n}{4}} > 0$.} Intuitively, this means that
the allocation is very non symmetric with big holes in the less loaded
bins. While $\Phi$ may sometimes grow in expectation, we will show that if that happens, then the asymmetry implies
that $\Gamma$ is dominated by $\Psi$ which decreases. Thus the
decrease in $\Psi$ offsets the increase in $\Phi$ and the expected
change in $\Gamma$ is negative.

Formally, if $\Ex[\Delta\Phi|x] \leq -\frac{\alpha\e}{4n}\Phi$, Lemma~\ref{lem:psidecreases} implies the result. Otherwise, by Lemma~\ref{lem:phibadcase} there are two subcases:

\medskip\noindent{\bf Case 2.1: $\Phi < \frac{\e}{4}\Psi$.} In this case, using Lemma~\ref{lem:psidecreases} and Corollary~\ref{cor:phiincrease}
\begin{equation*}
\Ex[\Delta\Gamma|x] = \Ex[\Delta\Phi|x]+\Ex[\Delta\Psi|x]  \leq \frac{2\alpha}{n}\Phi - \frac{\alpha \e}{n} \Psi + 1 \leq -\frac{\alpha \e}{2n} \Psi + 1 \leq -\frac{\alpha \e}{4n}\Gamma +1
\end{equation*}

\medskip\noindent{\bf Case 2.2: $\Gamma < cn$.} In this case, Corollaries~\ref{cor:phiincrease} and~\ref{cor:psiincrease} imply that
$$\Ex[\Delta \Gamma |x] \leq \frac{2\alpha}{n}\Gamma \leq 2c\alpha.$$
On the other hand, $c-\frac{\alpha\e}{4n} \Gamma \geq c(1-\frac{\alpha\e}{4})  > 2c\alpha$.

\medskip\noindent{\bf Case 3: $x_{\frac{3n}{4}} \leq x_{\frac{n}{4}} < 0$.} This case is similar to case 2. If $\Ex[\Delta\Psi|x] \leq -\frac{\alpha\e}{4n}\Psi$, Lemma~\ref{lem:phidecreases} implies the result. Otherwise, by Lemma~\ref{lem:psibadcase} there are two subcases:

\medskip\noindent{\bf Case 3.1: $\Psi < \frac{\e}{4}\Phi$.} In this case, using Lemma~\ref{lem:phidecreases} and Corollary~\ref{cor:psiincrease}, the claim follows.

\medskip\noindent{\bf Case 3.2: $\Gamma < cn$.} This case is the same as case

\end{proof}

\noindent Once we have shown that $\Gamma$ decreases in expectation when large, we can use that to bound the expected value of $\Gamma$.

We are now ready to prove Theorem \ref{thm:potentialbound}.
\begin{theorem} For any $t\geq 0$, $\Ex[\Gamma(t)] \leq \frac{4c}{\alpha\e} n$.
\end{theorem}
\begin{proof}
We show the claim by induction. For $t=0$, it is trivially true. By Theorem~\ref{thm:gammadecreases}, we have
\begin{align*}
\Ex[\Gamma(t+1)] &= \Ex[ \Ex[\Gamma(t+1)~|~\Gamma(t)]]\\
&\leq \Ex[ (1-\frac{\alpha\e}{4n})\Gamma(t) + c]\\
&\leq \frac{4c}{\alpha\e} n (1-\frac{\alpha \e}{4n}) + c\\
&\leq \frac{4c}{\alpha\e} n - c + c
\end{align*}
The claim follows.
\end{proof}

\end{document}